\documentclass[12pt]{article}
\usepackage[mathscr]{eucal}
\usepackage{amsmath,amsfonts,amssymb,amsthm}
\usepackage[matrix,arrow]{xy}
\usepackage{times}
\usepackage{float}
\usepackage{graphicx}
\usepackage{epstopdf}
\usepackage{multibox}
\usepackage{dcolumn}
\usepackage{indentfirst}

\advance\textheight\topskip \textwidth 36.5pc \oddsidemargin 16pt
\numberwithin{equation}{section} \makeatletter
\@addtoreset{equation}{section}

\newtheorem{prop}{Proposition}[section]

\title{A remark on the Lagrange structure
 of the unfolded field theory}
\date{}
\author{D. Kaparulin, S. Lyakhovich and A. Sharapov\\[2mm]
\small \emph{Department of Quantum Field Theory, Tomsk State University,}\\
\small \emph{Lenin ave. 36,  Tomsk 634050, Russia}\\
\small E-mails: dsc@phys.tsu.ru, sll@phys.tsu.ru,
sharapov@phys.tsu.ru}

\begin{document}
\maketitle

\begin{abstract}

Any local field theory can be equivalently reformulated in the
so-called unfolded form. General unfolded equations are
non-Lagrangian even though the original theory is Lagrangian.
Using the theory of a scalar field as a basic example, the concept
of Lagrange anchor is applied to perform a consistent
path-integral quantization of unfolded dynamics. It is shown that
the unfolded representation for the canonical Lagrange anchor of
the d'Alembert equation inevitably involves an infinite number of
space-time derivatives.

\end{abstract}

\section{Introduction}

Any system of local field equations can be equivalently
reformulated in the so-called unfolded form \cite{V0}, \cite{V1}.
The field content of the unfolded formulation includes, besides
the original fields, an infinite collection of auxiliary fields
absorbing all the space-time derivatives of the original fields.
The unfolded equations are generally not Lagrangian even though
the  system has been Lagrangian before unfolding. One of the
remarkable achievements of the unfolded formalism is that it is
the only known form of the interacting higher-spin field theories,
see \cite{V2}, \cite{BCIV} for reviews. On the other hand, the
absence of an action functional is usually regarded a serious
disadvantage of the unfolded dynamics, especially when the
quantization problem is the issue.

In \cite{KLS}, a new quantization method was proposed for not
necessarily   Lagrangian systems. The main geometric ingredient of
the method  is the notion of a Lagrange anchor. The existence of a
compatible Lagrange anchor appears to be a less restrictive
condition for the field equations than the requirement to admit a
variational formulation. Given a Lagrange anchor, the system can
be quantized in three different ways: (i) by converting the
original $d$-dimensional field theory into an equivalent
topological Lagrangian model in $d+1$ dimensions \cite{KLS}, (ii)
by constructing a generalized Schwinger-Dyson equation, which
defines the generating functional of Green's functions \cite{LS2},
(iii) by imbedding the original dynamics into an augmented
Lagrangian model in $d$ dimensions \cite{LS3}. The Lagrange anchor
also connects symmetries to conservation laws \cite{KLS2},
extending the Noether theorem beyond the class of Lagrangian
dynamics. Many explicit examples of Lagrange anchors are known for
various non-Lagrangian field theories \cite{KLS}-\cite{BG}.

In this paper we study the structure of Lagrange anchors
compatible with the unfolded form of dynamics. In particular, we
address the case where the unfolded field theory admits an
equivalent Lagrangian formulation with finite number of fields. In
this case, the general structure of the Lagrange anchor was first
established in \cite{KMaster}. In the present paper, to avoid
technicalities and to illuminate the main properties of the
unfolded  Lagrange anchors, we focus on the theory of a single
scalar field. The general constructions and conclusions, being
derived for this quite an elementary model, remain applicable to a
much broader class of dynamics.

The paper is organized as follows. The next section starts with
recalling the free-diffe\-ren\-tial-algebra approach to the
formulation of classical field theory and its relation to the
unfolded formalism. Then we briefly describe the on-shell and
off-shell unfolded representations for the d'Alembert equation.
Section 3 contains a brief exposition of basic definitions,
constructions and motivations concerning the Lagrange anchor and
its application to the path-integral quantization of
(non-)Lagrangian field theories. In Section 4, we identify a
natural Lagrange anchor for the unfolded representation of the
d'Alembert equation. It is the anchor which provides the
equivalence of the quantized unfolded model to the standard
path-integral quantization of the scalar field. In Section 5, we
prove a no-go  theorem for the existence of Lagrange anchors
without space-time derivatives in the on-shell unfolded formalism.
In Section 6, we summarize the results and discuss possible
applications of the proposed Lagrange anchor construction to the
unfolded form of field theory.

\section{Unfolded representation of the d'Alembert equation}

We start with a very brief review of the free-differential-algebra
(FDA) approach to  the formulation of general covariant field
theories \cite{AF}, \cite{FG}. In the FDA formalism one deals with
a multiplet  $\{w^a\}$ of differential forms of various degrees
on the space-time manifold $M$. The field equations are of the
form
\begin{equation}\label{Q}
d w^a=Q^a(w)\,.
\end{equation}
The r.h.s. of these equations, $Q$'s, are given by exterior
polynomials in $w$'s. The identity $d^2=0$ for the exterior
differential results in the following compatibility conditions for
equations (\ref{Q}):
\begin{equation}\label{CC}
    Q^a\frac{\partial Q^b}{\partial w^a}=0\,.
\end{equation}
These conditions are assumed to be identically satisfied for all
$w$'s. Equation  (\ref{CC}) has  the following geometric
interpretation. The forms $w^a$ can be treated as coordinates on
an $\mathbb{N}$-graded manifold $\mathcal{M}$ with the degrees of
coordinates being the corresponding form-degrees. Then the r.h.s
of (\ref{Q}) defines a \textit{homological vector field}
$Q=Q^a\partial_a$ on $\mathcal{M}$, i.e., a smooth vector field of
degree $1$ squaring to zero (\ref{CC}). Besides an explicit
invariance under the diffeomorphisms of $M$, equations (\ref{Q})
enjoy the gauge symmetries
\begin{equation}\label{GT}
    \delta_\varepsilon w^a=d\varepsilon^a-\frac{\partial Q^a}{\partial
    w^b}\varepsilon^b\,,
\end{equation}
where the gauge parameters $\varepsilon^a$ are the forms of
appropriate degrees.

In the field theory, it can be useful to consider non-free
differential algebras, when the differential equations (\ref{Q})
are supplemented  by a set of algebraic constraints
\begin{equation}\label{t}
    T^A(w)=0\,.
\end{equation}
Being exterior  polynomials in $w$'s, the constraints generate an
ideal $J\subset \Lambda(M)$ in the algebra of all differential
forms on $M$. To avoid further compatibility conditions, $J$ is
assumed to be a differential ideal, i.e., $dJ\subset J$. On
account of (\ref{Q}), the last requirement amounts   to
\begin{equation}\label{DI}
Q^a\frac{\partial T^A}{\partial w^a}= U^A_B T^B
\end{equation}
for some exterior polynomials $U^A_B(w)$. Under certain regularity
conditions, the algebraic constraints (\ref{t}) define a smooth
submanifold $\mathcal{N}\subset \mathcal{M}$ and condition
(\ref{DI}) says that the homological vector field $Q$ is tangent
to $\mathcal{N}$. Restricting the vector field $Q$ to
$\mathcal{N}$, e.g. by explicitly solving the algebraic
constraints, we arrive to another FDA with fewer generators.

By unfolding (or unfolded representation of) a local field theory,
one usually means a reformulation of the original field equations
in the form (\ref{Q}), (\ref{t}). Such a reformulation usually
involves an infinite number of auxiliary fields. In what follows
we will distinguish between \textit{off-shell} and
\textit{on-shell} unfolded representations. In the former case,
there are  both the differential equations (\ref{Q}) \textit{and}
the algebraic constraints  (\ref{t}) (also called shell
conditions). In the on-shell formulation, the fields $w$ are
unconstrained, that can be a result of solving the algebraic
constraints of the off-shell formulation.

In this paper we will mostly consider the unfolded representation
for the free, massless, scalar field theory in $d$-dimensional
Minkowski space, though our constructions and conclusions remain
applicable to a much broader class of field theories. The standard
formulation of the scalar field theory is based on the d'Alembert
equation
\begin{equation}\label{DA}
    \Box \varphi=0\,.
\end{equation}

Loosely, the idea behind constructing the unfolded representation
(\ref{Q}), (\ref{t}) for a given system of differential equations
is to introduce an infinite collection of auxiliary fields
absorbing all the derivatives of the original fields. This
procedure is known as the infinite jet prolongation \cite{EDS}. In
terms of the extended set of fields, the system of differential
equations (\ref{Q}) appears as the definition of the jet
prolongation (this is called a contact system), whereas the
original field equations along with all their differential
consequences turn into the algebraic constraints (\ref{t}) in the
jet space. To keep the system explicitly invariant under
diffeomorphisms, the unfolded representation is formulated in a
general frame, even though the space-time $M$ can be flat (as is
supposed throughout this paper). As usual, this implies
introduction of a vielbein $e^a\in \Lambda^1(M)$ and a compatible
Lorentz connection $\omega^{ab}=-\omega^{ba}\in \Lambda^{1}(M)$.
All the Lorentz indices $a,b,c, ...$ are raised and lowered with
Minkowski metric $\eta_{ab}$.

The unfolded version of the d'Alembert equation (\ref{DA}) reads
\cite{BCIV}, \cite{SV}:
\begin{equation}\label{ew}
De^a=0\,,\qquad d\omega^a{}_b=-\omega^a{}_c\wedge\omega^c{}_b\,,
\end{equation}
\begin{equation}\label{phi}
D\varphi_{a_1\cdots a_s}=e^{a}\varphi_{aa_1\cdots a_s}\,,\qquad
s=0,1,\ldots\,,
\end{equation}
\begin{equation}\label{tphi}
    \varphi{}^a{}_{a a_1\cdots a_s}=0\,,\qquad s=0,1,\ldots\,.
\end{equation}
Here $D=e^a\wedge \nabla_a=d+\omega$ is the Lorenz-covariant
differential and $\{\varphi_{a_1\cdots a_s}\}$ is an infinite
collection of fully symmetric Lorentz tensors.

Let us comment on the structure of the unfolded equations. The
first equation in (\ref{ew}) is the usual zero-torsion condition.
It allows one to express the components of the Lorentz connection
$\omega^{ab}$ via the vielbein field $e^a$. The second equation in
(\ref{ew}) is the zero-curvature condition for the Lorentz
connection. So, there is a coordinate system $\{x^a\}$ on $M$ in
which $\omega^{ab}=0$ and $e^a=dx^a$. This also means that the
1-form fields $\omega^{ab}$ and $e^a$ are pure gauge. Equations
(\ref{phi}) can be viewed as a reparametrization invariant form of
the \emph{contact system} \cite{EDS}. By taking linear
combinations these equations can be rearranged into the form
\begin{equation}\label{tt}
    \varphi_{a_1\ldots
    a_s}=\nabla_{a_1}\cdots\nabla_{a_s}\varphi\,.
\end{equation}
So, the contact equations just define the auxiliary fields
$\varphi_{a_1\cdots a_s},\, s>0$, in terms of successive
derivatives (infinite jet) of a single scalar field $\varphi$.

Substituting (\ref{tt}) into (\ref{tphi}), we get the sequence of
equations
\begin{equation}\label{}
    \nabla_{a_1}\cdots \nabla_{a_s}\Box \varphi=0\,,\qquad
    \Box=\nabla^a\nabla_a\,, \qquad s=0,1,2, \ldots ,
\end{equation}
which is clearly equivalent to the original d'Alembert's equation
(\ref{DA}). The unfolded system is obviously consistent, because
applying $d$ to both sides of (\ref{ew}), (\ref{phi}) and
(\ref{tphi}) does not lead to any new condition. The differential
algebra underlying the unfolded formulation above is generated by
the finite set of 1-forms $\{e^a$, $\omega^{ab}\}$ and the
infinite collection of 0-forms $\{\varphi_{a_1\cdots a_s}\}$. In
view of the algebraic constraints (\ref{tphi}), it is not a free
algebra. To pass from the off-shell to on-shell formulation  one
only has to assume the Lorentz tensors $\varphi_{a_1\cdots a_s}$
to be traceless.

For the field equations (\ref{ew})-(\ref{tphi}) the general gauge
transformations (\ref{GT}) read
\begin{equation}\label{GSvLc}
\begin{array}{l}
\displaystyle \delta_\varepsilon \varphi_{a_1\cdots
a_s}=\varepsilon^{a}\varphi_{aa_{1}\cdots
a_{s}}-\sum_{i=1}^{s}\varepsilon_{a_i}{}^{a}\varphi_{a_1\cdots a_{i-1}a a_{i+1}\cdots a_s}\,,\\[5mm]
\displaystyle \delta_\varepsilon
e^{a}=D\varepsilon^a-\varepsilon^{ab}e_b,\qquad \delta_\varepsilon
\omega^{ab}=D\varepsilon^{ab}\,.
\end{array}\end{equation}
The gauge parameters $\varepsilon^a$ and
$\varepsilon^{ab}=-\varepsilon^{ba}$ correspond, respectively, to
the general coordinate transformations and the local Lorentz
rotations.

Notice that the subsystem  (\ref{ew}), (\ref{phi}), being
considered separately from constraints (\ref{tphi}), remains
formally consistent. This truncated system, however, is
dynamically empty: It just expresses the higher rank Lorentz
tensors $\varphi_{a_1\cdots a_s}$ in terms of the unconstrained
scalar field $\varphi$, and its derivatives, according to
({\ref{tt}}). As a result, the truncated system of equations
possesses the extra gauge symmetry of the form
\begin{equation}\label{gtr}
    \delta_{\varepsilon}\varphi_{a_1\cdots
    a_s}=\nabla_{a_1}\cdots\nabla_{a_s}\varepsilon\,.
\end{equation}
Combining these transformations  with (\ref{GSvLc}), one can gauge
out all the fields. Also notice that the transformations
(\ref{gtr}) involve $s$ derivatives of the gauge parameter for the
field $\varphi_{a_1\cdots a_s}$. From this viewpoint, these
transformations are local, as the number of derivatives is finite
for every single field. However, the order of the derivatives is
unbounded, because it is growing with the order of the jet,
whereas the jet prolongation is infinite in the unfolded
formalism. We will return to this point in Section 4. The
algebraic condition (\ref{tphi}) breaks the gauge symmetry
(\ref{gtr}), so the unfolded system gets back its physical degrees
of freedom and becomes equivalent  to the d'Alembert equation as
soon as the Lorentz tensors $\varphi_{a_1\cdots a_s}$ are
constrained to be traceless.

It is  worth noting that the unfolded system
(\ref{ew})-(\ref{tphi}) is not Lagrangian in any dimension unlike
the d'Alembert equation (\ref{DA}). Furthermore, it remains
non-Lagrangian even if one omits equations (\ref{ew}),
(\ref{tphi}) by imposing gauge fixing conditions on the vielbein
and the Lorentz connection and assuming the auxiliary fields
$\varphi_{a_1\cdots a_s}$ to be traceless. This fact may seem
restricting the use of the unfolded formalism especially when it
comes to quantizing the classical fields. Actually, the Lagrangian
formalism is a special case of a more general concept called a
Lagrange structure \cite{KLS}. The existence of a compatible
Lagrange structure is much less restrictive for the field
equations than the requirement to admit a Lagrangian. Given a
Lagrange structure, it is still possible to consistently quantize
the classical theory \cite{KLS}, \cite{LS2}, \cite{LS3}, even
though it admits no Lagrangian. The Lagrange structure also
connects symmetries with conservation laws \cite{KLS2}. The next
section contains a simplified exposition of the Lagrange structure
and the corresponding quantization method. This description is
sufficient for many field-theoretical applications, including the
issues of this paper. A more systematic and rigorous presentation
can be found in \cite{KLS}-\cite{KLS2}.

\section{A generalized Schwinger-Dyson equation}

In the covariant formulation of quantum field theory one usually
studies the path integrals of the form
\begin{equation}\label{PI}
    \langle \mathcal{O}\rangle =\int [d\varphi]
    \,\mathcal{O}[\varphi]\,e^{\frac i{\hbar}S[\varphi]}\,.
\end{equation}
After normalization, the integral defines the quantum average of
an observable $\mathcal{O}[\varphi]$ in the theory with action
$S[\varphi]$.  It is believed that evaluating the path integrals
for various observables $\mathcal{O}$, one can extract all
physically relevant information about the quantum dynamics of the
model.

The  functional $\Psi [\varphi]=e^{\frac i{\hbar}S[\varphi]}$,
weighting the contribution of a particular field configuration
$\varphi$ to the quantum average, is known as the Feynman
probability amplitude on the configuration space of fields. This
amplitude can be defined as a unique (up to a normalization
factor) solution to the Schwinger-Dyson (SD)
equation\footnote{Here we use the condensed notation, so that the
partial derivatives with respect to fields should be understood as
variational ones.}
\begin{equation}\label{SDf}
   \left (\frac{\partial S}{\partial\varphi^i}+i\hbar\frac{\partial}{\partial
    \varphi^i}\right) \Psi[\varphi]=0\,.
\end{equation}
Performing the Fourier transform from the fields $\varphi$ to
their sources $\bar\varphi$, we can bring (\ref{SDf}) to a more
familiar form
\begin{equation}\label{SD}
\left( \frac{\partial
S}{\partial\varphi^i}(\widehat{\varphi})-\bar\varphi_i\right)Z[\bar\varphi]=0\,,\qquad
\widehat{\varphi}{}^i\equiv i\hbar\frac{\partial}{\partial
\bar\varphi_i}\,,
\end{equation}
where
\begin{equation}\label{Z}
    Z[\bar\varphi]=\int [d\varphi] e^{\frac{i}{\hbar}(S-\bar\varphi\varphi)}
\end{equation}
is the generating functional of Green's functions.

The following observations provide guidelines for the
generalization of the Schwinger-Dyson equation to non-Lagrangian
field theory.

\vspace{2mm}

    $(i)$ Although the Feynman probability amplitude involves an action functional,
    the SD  equations (\ref{SDf})
    contain solely the classical field equations, not the action as such.

\vspace{2mm}

    $(ii)$ In the classical limit $\hbar\rightarrow 0$, the second
    term in the SD equation (\ref{SDf}) vanishes, and the
    Feynman probability amplitude $\Psi$ turns into the Dirac distribution supported
    at the classical solutions to the field equations. Formally, $\Psi[\varphi]|_{\hbar\rightarrow 0} \sim
    \delta[\partial_i
    S]$ and one can think of the last expression as the classical probability amplitude.

\vspace{2mm}

    $(iii)$ It is quite natural to treat the sources $\bar\varphi$ as the momenta
    canonically conjugate  to the fields $\varphi$,
    so that the only non-vanishing Poisson brackets are
    $\{\varphi^i,\bar\varphi_j\}=\delta^i_j$.
    Then, one can regard the SD operators
\begin{equation}\label{SDOP}
\frac{\partial
S}{\partial\varphi^i}+i\hbar\frac{\partial}{\partial
    \varphi^i}
\end{equation}
involved in (\ref{SDf})  as resulting from the canonical
quantization of the first class constraints
\begin{equation}\label{CLA}
    \partial_i S[\varphi]-\bar{\varphi}_i\approx 0
    \end{equation}
    on the phase space of fields and sources. Upon this
    interpretation, the Feynman probability amplitude describes  a unique physical state
    of a first-class constrained theory.
    This state is unique because  the ``number'' of the first class constraints (\ref{CLA}) is equal  to the
    ``dimension'' of the configuration space of fields. Quantizing the constrained system (\ref{CLA}) in
    the   momentum representation yields  the SD equation (\ref{SD}) for
    the partition function $Z[\bar\varphi]$.

\vspace{2mm}

The above interpretation of the SD equations as operator first
class constraints on a physical wave-function suggests a direct
way to their generalization. Consider a set of field equations
\begin{equation}\label{T-eq}
T_a(\varphi^i)=0\,,
\end{equation}
which do not necessarily follow from the variational principle. In
this case, the (discrete parts of) superindices $a$ and $i$ may
run over different sets. Proceeding from the heuristic arguments
above, we can take the following ansatz for the $\varphi
\bar\varphi$-symbols of the Schwinger-Dyson operators:
\begin{equation}\label{TT}
\mathcal{T}_a=T_a(\varphi)-V_a^i(\varphi)\bar\varphi_i+
O(\bar\varphi^2)\,.
\end{equation}
The symbols are defined as formal power series in  sources
$\bar\varphi$ with leading terms being the classical equations of
motion. Requiring the Hamiltonian constraints
$\mathcal{T}_a\approx 0$ to be first class, i.e.,
\begin{equation}\label{inv}
    \{\mathcal{T}_a, \mathcal{T}_b\}=U_{ab}^c \mathcal{T}_c \,,\qquad
    U_{ab}^c(\varphi,\bar\varphi)=C^c_{ab}(\varphi)+ O(\bar\varphi)\,,
\end{equation}
we obtain an infinite set of relations on the expansion
coefficients of $\mathcal{T}_a$ in the powers of sources. In
particular, verifying the involution relations (\ref{inv}) up to
zero order in $\bar\varphi$, we find
\begin{equation}\label{anchor}
    V_a^i\partial_iT_b-V_b^i\partial_iT_a=C_{ab}^c T_c\,.
\end{equation}
The value $V_a^i(\varphi)$ defined by (\ref{anchor}) is called the
\emph{Lagrange anchor}. The entire sequence of the expansion
coefficients defines the \emph{Lagrange structure}.

For variational field equations, $T_a=\partial_i S$, one can set
the Lagrange anchor to be the unit matrix $V_a^i=\delta^i_a$. This
choice results in the standard Schwinger-Dyson operators
(\ref{SDOP}) obeying the abelian involution relations. Generally,
the Lagrange anchor may be field-dependent and/or noninvertible.
If the Lagrange anchor is invertible (in which case the number of
equations must coincide with the number of fields), then the
operator $V^{-1}$ plays the role of integrating multiplier in the
inverse problem of calculus of variations. So, the existence of
the invertible Lagrange anchor amounts to the existence of action.
The other extreme choice, $V=0$, is always possible and
corresponds to the classical probability amplitude
$\Psi[\varphi]\sim\delta[T_a(\varphi)]$ supported at the classical
solutions.

In the non-Lagrangian case, the constraints (\ref{TT}) are not
generally the whole story. The point is that the number of
(independent) field equations can happen to be less than the
dimension of the configuration space of fields. In that case, the
field equations (\ref{T-eq})  do not specify a unique solution
with prescribed boundary conditions or, stated differently, the
system enjoys a gauge symmetry  generated by an on-shell
integrable vector distribution
$R_\alpha=R_\alpha^i(\varphi)\partial_i$ such that
\begin{equation}\label{}
R_\alpha^i\partial_i T_a=U_{\alpha a}^b T_b \,,\qquad [R_\alpha,
R_\beta]=U_{\alpha\beta}^\gamma R_\gamma + T_a
U_{\alpha\beta}^{ai}\partial_i
\end{equation}
for some structure functions $U^b_{\alpha a}(\varphi)$ and
$U_{\alpha\beta}^{ai}(\varphi)$.
 To take the gauge invariance into account at quantum level, one has to
 impose additional first class constraints on the fields and sources. Namely,
\begin{equation}\label{R}
    \mathcal{R}_\alpha =R_\alpha^i(\varphi)\bar\varphi_i+O(\bar\varphi^2)\approx 0\,.
\end{equation}
The leading terms of these constraints coincide with the  $\varphi
\bar\varphi$-symbols of the gauge symmetry generators and the
higher orders in $\bar\varphi$ are determined from the requirement
the Hamiltonian constraints $\mathbb{T}_I=(\mathcal{T}_a,
\mathcal{R}_\alpha)$ to be in involution \footnote{For a
Lagrangian gauge theory we have
$\mathcal{T}_i=\partial_iS-\bar\varphi_i$ and $\mathcal{R}_\alpha
=-R^i_\alpha \mathcal{T}_i=R_\alpha^i \bar\varphi_i$. In this
case, one may omit the ``gauge'' constraints
$\mathcal{R}_\alpha\approx 0$ as they are given by linear
combinations of the ``dynamical'' constraints
$\mathcal{T}_i\approx 0$.}. With all the gauge symmetries
included, the constraint surface $\mathbb{T}_I\approx 0$ is proved
to be a Lagrangian submanifold in the phase space of fields and
sources and the gauge invariant probability amplitude is defined
as a unique solution to the \textit{generalized SD equation}
\begin{equation}\label{SDE}
    {\widehat{\mathbb{T}}}{}_I\Psi=0\,.
\end{equation}
The last formula is just the definition of a physical state in the
Dirac quantization method of constrained dynamics. A more
systematic treatment of the generalized SD equation within the
BRST formalism can be found in \cite{KLS}, \cite{LS2}, \cite{LS3}.

In what follows we will refer to the first class constraints
$\mathbb{T}_I \approx 0$ as the \textit{Schwinger-Dyson extension}
of the original equations of motion (\ref{T-eq}). Notice that the
defining relations (\ref{anchor}) for the Lagrange anchor together
with the ``boundary conditions'' (\ref{TT}) and (\ref{R}) do not
specify a unique SD extension for a given system of field
equations. One part of the ambiguity is related to canonical
transformations in the phase space of fields and sources. If the
generator $G$ of  a canonical transform is at least quadratic in
sources,
\begin{equation}\label{}
G=\frac12 G^{ij}(\varphi)\bar\varphi_i\bar\varphi_j +
O(\bar\varphi^3)\,,
\end{equation}
then the transformed constraints
\begin{equation}\label{G-tr}
\begin{array}{l}
    \mathcal{T}'_a=e^{\{G,\,\,\cdot\,\,\}}\mathcal{T}_a=T_a+ (V_a^i
    +G^{ij}\partial_jT_a)\bar\varphi_i + O(\bar\varphi^2)\,,\\[5mm]
    \mathcal{R}'_\alpha =
    e^{\{G,\,\,\cdot\,\,\}}\mathcal{R}_\alpha=R_\alpha^i\bar\varphi_i+O(\bar\varphi^2)
\end{array}
\end{equation}
are in involution and start  with the same equations of motion and
gauge symmetry generators. Another ambiguity stems from changing
the basis of the constraints:
\begin{equation}\label{U-tr}
\begin{array}{l}
    \mathcal{T}''_a=U^b_a\mathcal{T}_b +U_a^\alpha \mathcal{R}_\alpha = {T}_a +(V_a^i+ A_a^{bi}T_b+B_a^\alpha
    R_\alpha^i)\bar\varphi_i+O(\bar\varphi^2)\,,\\[5mm]
    \mathcal{R}''_\alpha=U_\alpha^\beta
    \mathcal{R}_\beta+U_\alpha^a
    \mathcal{T}_a=R_\alpha^i\bar\varphi+O(\bar\varphi^2)\,,
    \end{array}
\end{equation}
where
\begin{equation}\label{}
\begin{array}{ll}
    U^b_a=\delta_a^b+A_a^{bi}\bar\varphi_i+O(\bar\varphi^2)\,,&\qquad
    U_a^\alpha=B_a^\alpha + O(\bar\varphi)\,,\\[5mm]
    U_\alpha^\beta=\delta_\alpha^\beta+O(\bar\varphi)\,,&\qquad
    U_\alpha^a=O(\bar\varphi)\,.
    \end{array}
\end{equation}
Combining (\ref{G-tr}) with (\ref{U-tr}) we see that the Lagrange
anchor is defined modulo the equivalence relations
\begin{equation}\label{eqw}
V_a^i \sim V_a^i+ T_b A_a^{bi}+B_a^\alpha R_\alpha^i
+G^{ij}\partial_jT_a\,.
\end{equation}
A more rigorous treatment of the Lagrange structure and
generalized Schwinger-Dyson equation is provided by the
corresponding BRST formalism \cite{KLS}, \cite{LS2}. From the
viewpoint of the  BRST theory, all transformations (\ref{G-tr})
and (\ref{U-tr}) are induced by canonical transforms in the
ghost-extended phase space, and the equivalence classes of
Lagrange anchors are identified with certain classes of the BRST
cohomology.

\section{The Schwinger-Dyson extension of unfolded dynamics}

In this section, the general quantization procedure described
above is applied to the scalar field theory in the unfolded
formulation.

Imposing the requirements of (i) space-time locality, (ii)
Poincar\'e covariance, and (iii) linearity in fields and sources,
one can see that the most general SD extension of the original
d'Alembert's equation (\ref{DA}) reads
\begin{equation}\label{SDSc}
    \Box \varphi +\sum_{n=0}^Na_n\Box^n \bar\varphi\approx 0\,,
\end{equation}
where $a_n$ are constants. These constraints  are clearly in
abelian involution. The local canonical transform
\begin{equation}\label{}
    \varphi \rightarrow \varphi - \sum_{n=1}^N
    a_n\Box^{n-1}\bar\varphi\,,\qquad \bar\varphi \rightarrow
    \bar\varphi\,,
\end{equation}
brings the first class constraints (\ref{SDSc}) to the form
\begin{equation}\label{CLASF}
    \Box \varphi +a_0 \bar\varphi\approx 0\,.
\end{equation}
This means that the most general Lagrange anchor (\ref{SDSc}) is
equivalent (up to normalization) to the canonical one (\ref{CLA}),
whenever $a_0\neq 0$. If $a_0=0$ the anchor is trivial.

Unlike the d'Alembert equation, the unfolded system
(\ref{ew})-(\ref{tphi}) is non-Lagrangian. Therefore, one cannot
use the conventional quantization prescriptions. The unfolded
equations, however, admit a quite natural SD extension, which can
be viewed as an ``unfoldization'' of the constraints
(\ref{CLASF}). To describe this extension, following the general
procedure of Section 3, we introduce the sources $\bar e_a$ and
$\bar\omega_{ab}=-\bar\omega_{ba}$ for the vielbein and connection
fields to be $(d-1)$-forms on $M$. The fully symmetric Lorentz
tensors $\bar\varphi^{a_1\cdots a_s}$ with values in $d$-forms are
introduced as the sources for $\varphi_{a_1\cdots a_s}$. The
Poisson brackets in the phase space of fields and sources are
defined by the following symplectic form:
\begin{equation}\label{}
\Omega=\int_M \big(\delta e^a\wedge  \delta \bar e_a +\delta
\omega^{ab}\wedge  \delta \bar\omega_{ab}+ \sum_{s=0}^\infty
\delta\varphi_{a_1\cdots a_s}\wedge \delta\bar\varphi^{a_1\cdots
a_s}\big)\,.
\end{equation}
Here the sign $\wedge$ stands for both the exterior product of
differential forms on $M$ and the exterior product of variational
differentials.

Let us first consider the off-shell formulation. In this case, the
differential equations (\ref{ew}) and (\ref{phi}) remain intact,
while the algebraic constraints (\ref{tphi}) are replaced by the
expressions
\begin{equation}\label{dtphi}
    \varphi^a{}_{aa_1\cdots a_s}+ \nabla_{a_1}\cdots \nabla_{a_s} R\approx
    0\,,
\end{equation}
where
\begin{equation}\label{r}
R=\ast\sum_{n=0}^{\infty} (-1)^n \nabla_{a_1}\cdots
\nabla_{a_n}\bar\varphi^{a_1\cdots a_n}\,,
\end{equation}
and the Hodge operator $\ast: \Lambda^p(M)\rightarrow
\Lambda^{d-p}(M)$ is defined with respect to the space-time
metric. Since the unfolded system (\ref{ew})-(\ref{tphi}) is gauge
invariant we should also add the constraints associated to the
gauge symmetry generators (\ref{GSvLc}). These are given by
\begin{equation}\label{RR}
\begin{array}{l}
\displaystyle
    D \bar e_a-\sum_{s=0}^\infty \varphi_{aa_1\cdots
    a_s}\bar\varphi^{a_1\cdots a_s}\approx
    0\,,\\[5mm]\displaystyle
    D\bar
    \omega^{ab}-e^{[a}\wedge \bar e^{b]}+\sum_{s=0}^\infty (s+1)\bar\varphi^{a_1\cdots
    a_s[a}\varphi^{b]}{}_{a_1\cdots a_s}\approx 0\,.
    \end{array}
\end{equation}

It is clear that the Hamiltonian constraints (\ref{ew}),
(\ref{phi}), and (\ref{RR}) are of the first class. Then, it
remains to prove the involution of constraints (\ref{dtphi}) among
themselves and with the other constraints. Observe that the SD
constraints (\ref{dtphi}) are explicitly covariant under the
action of diffeomorphisms and the local Lorentz transformations
generated by  (\ref{GSvLc}). So, the Poisson brackets of
(\ref{RR}) and (\ref{dtphi}) must be proportional to
(\ref{dtphi}). To prove  the involution of (\ref{ew}), (\ref{phi})
and (\ref{dtphi}), we notice that the Hamiltonian flow
$\{R,\,\cdot\,\}$ generates the gauge transformations (\ref{gtr})
for the truncated system (\ref{ew}), (\ref{phi}). Hence, the
constraints (\ref{dtphi}) Poisson commute with (\ref{ew}) and
(\ref{phi}). Finally, we know that the first class constraints
(\ref{phi}) are equivalent to (\ref{tt}). Taking into account
(\ref{tt}), we can replace (\ref{dtphi}) with the equivalent set
of constraints
\begin{equation}\label{ttt}
    \nabla_{a_1}\cdots \nabla_{a_s}\left((\Box \varphi
    +\ast \bar\varphi)+\ast\sum_{n=1}^\infty (-1)^n\nabla_{a_1}\cdots
    \nabla_{a_n}\bar\varphi^{a_1\cdots a_n}\right)\approx 0\,.
\end{equation}
The abelian involution of the constraints (\ref{ttt}), and hence
(\ref{dtphi}), is now obvious.

Let us summarize the interim result. The system of equations
(\ref{ew}), (\ref{phi}), (\ref{dtphi}), (\ref{RR}) defines the
complete set of the Schwinger-Dyson constraints for the off-shell
unfolded formulation of scalar field's dynamics.  Below,  we will
see that the corresponding probability amplitude leads to the
quantum theory, which is completely equivalent to the conventional
theory based on the d'Alembert equation.

Two steps are needed to put the unfolded system (\ref{ew}),
(\ref{phi}), (\ref{dtphi}), (\ref{RR}) on-shell. First, we express
the traces of the tensor fields $\varphi_{a_1\cdots a_s}$ via the
sources according to (\ref{dtphi}) and substitute them into the
constraints (\ref{ew}), (\ref{phi}), (\ref{RR}). Constraints
(\ref{dtphi}) are now identically satisfied, and the remaining SD
constraints are still in involution. As  the second step, we set
the trace part of the sources $\bar\varphi^{a_1\cdots a_s}$ to
zero. Again, this cannot break involution as the constraints
derived at the first step do not involve the traces of the fields
$\varphi_{a_1\cdots a_s}$. Denote by ${\phi}_{a_1\cdots a_s}$ and
${\bar\phi}{}^{a_1\cdots a_s}$ the traceless parts of the tensors
$\varphi_{a_1\cdots a_s}$ and $\bar\varphi^{a_1\cdots a_s}$. Then
we arrive at the following SD constraints:
\begin{equation}\label{USF-on-An}
  De^a\approx 0\,,\qquad d\omega^a{}_b+\omega^a{}_c\wedge\omega^c{}_b\approx 0\,,\qquad D\Phi_{a_1\cdots a_s}-e^a\Phi_{a a_1\cdots a_s}\approx
0\,,
\end{equation}
\begin{equation}\label{GON1}
    D \bar e_a-\sum_{s=0}^\infty \Phi_{aa_1\cdots
    a_s}\bar\phi^{a_1\cdots a_s}\approx 0\,,
\end{equation}
\begin{equation}\label{GON3}
    D\bar
    \omega^{ab}-e^{[a}\wedge \bar e^{b]}+\sum_{s=0}^\infty(s+1)\bar\phi^{a_1\cdots
    a_s[a}\phi^{b]}{}_{a_1\cdots a_s}\approx 0\,,
\end{equation}
where
\begin{equation}\label{BPHI}
\Phi_{a_1\cdots a_s}={\phi}_{a_1\cdots
a_s}-\sum_{k=1}^{[s/2]}\alpha^s_{k}\eta_{(a_1 a_2}\cdots
\eta_{a_{2k-1}a_{2k}}\nabla_{a_{2k+1}}\cdots\nabla_{a_s)}\Box^{k-1}\widetilde{R}\,,
\end{equation}
\begin{equation}\label{}
    \widetilde{R}=\ast\sum_{k=0}^\infty(-1)^k\nabla_{a_1}\cdots
    \nabla_{a_k}\bar\phi^{a_1\cdots a_k}\,,
\end{equation}
and  the coefficients $\alpha^{s}_{k}$ are determined  by the
recurrent relations
\begin{equation}
\alpha^s_1=\frac1{d+2s-4}\,, \qquad
\alpha^s_k=\frac{\alpha^s_{k-1}}{2k+2-d-2s}\,.
\end{equation}
The round brackets in (\ref{BPHI}) mean summation over all
inequivalent permutations of the indices enclosed.   Notice that
the SD constraints (\ref{GON1}) involve bilinear combinations of
sources, cf. (\ref{R}). On account of (\ref{USF-on-An}) we can
replace (\ref{GON1} ) with the equivalent constraints
\begin{equation}\label{GON2}
    D\bar e_a-\sum_{s=0}^\infty \bar\phi^{a_1\cdots a_s}\nabla_a\phi_{a_1\cdots
    a_s}\approx 0\,.
\end{equation}
Equations (\ref{USF-on-An}), (\ref{GON3}), (\ref{GON2}) define an
equivalent basis of the SD constraints that are at most linear in
sources. These relations define the Lagrange structure for the
on-shell unfolded formulation of the scalar field.

From the viewpoint  of the constrained Hamiltonian dynamics,  the
transition from the off-shell to the on-shell formulation is the
Hamiltonian reduction by the second class constraints
\begin{equation}
\bar\varphi^{a}_{\phantom{a}aa_{1}\ldots a_s}\approx 0,\qquad
\varphi^{a}_{\phantom{a}aa_1\cdots
a_s}+\nabla_{a_1}\cdots\nabla_{a_{s}}R\approx 0\,,
\end{equation}
and the canonical Poisson brackets of the traceless fields and
sources $\{\phi_{a_1\cdots a_s},\bar\phi^{b_1\cdots b_k}\}$ appear
as Dirac's brackets in the reduced phase space.

Let us comment on the unusual property of the SD constraints
derived above both in the off-shell and on-shell formulations. The
unfolded field equations involve infinitely many fields, as the
jet prolongation is infinite in this formalism. The SD extensions
(\ref{USF-on-An}), (\ref{GON3}), (\ref{GON2}) of the field
equations involve a finite number of derivatives for every single
field and source, though the order of derivatives increases with
the order of jet. This fact directly follows from the structure of
the Hamiltonian generator (\ref{r}) for the gauge transformations
(\ref{gtr}) of the contact system. As we will see in the next
section this property is unavoidable for the Lagrange anchor in
the unfolded formalism.

Let us check that the quantization of the unfolded system results
in the standard Feynman's probability amplitude for the free
massless scalar field. Given the off-shell constraints (\ref{ew}),
(\ref{phi}), (\ref{dtphi}), (\ref{RR}), the generalized SD
equations for the quantum probability amplitude (\ref{SDE}) are
readily constructed following the general prescriptions of Section
3. They are satisfied by the functional
\begin{equation}\label{PAUSF}
\Psi\sim e^{\frac i\hbar S[\varphi]}\delta(de+\omega\wedge
e)\delta(d\omega+\omega\wedge\omega)
\prod_{s=0}^\infty\delta(D\varphi_{a_1\cdots
a_{s}}-e^a\varphi_{aa_1\cdots a_s})\,,
\end{equation}
where
\begin{equation}\label{S}
S[\varphi]=\frac{1}{2}\int_M d\varphi\wedge \ast d\varphi
\end{equation}
is the usual action for the free massless scalar field in general
gravitational background.

As is seen, no quantum fluctuations arise for the holonomic
constraints, i.e., the equations without source extensions. With
the $\delta$-functions of the contact equations all the auxiliary
fields $\varphi_{a_1\cdots a_s}$  can be integrated out in the
path integral. The vielbein and the Lorentz connection also remain
classical fields with no quantum fluctuations. As they are pure
gauge, their contribution is eliminated from the path integral by
imposing gauge-fixing conditions. This yields  the standard
Feynman's amplitude for the original scalar field. The fact that
the Schwinger-Dyson equations define a ``smeared'' probability
amplitude (rather than a $\delta$-distribution) confirms that the
constructed Lagrange anchor is nontrivial.

\section{Algebraic Lagrange anchors}

The on-shell unfolded representation for the dynamics of scalar
field (\ref{ew}), (\ref{phi}) suggests to consider a general FDA
generated by finite or countable sets of $1$-forms $\{\theta^a\}$
and $0$-forms $\{\phi^i\}$. With this field content, the general
FDA equations (\ref{Q}) take the form
\begin{equation}\label{FDA}
    d\phi^i=A^i_a(\phi)\theta^a\,,\qquad
    d\theta^a=C^a_{bc}(\phi)\theta^b\wedge\theta^c\,,
\end{equation}
where $A$'s and $C$'s are smooth functions of scalar fields
$\phi^i$. In view of (\ref{CC}) these structure functions obey the
relations
\begin{equation}\label{CCC}
[A_a,A_b]=C_{ab}^c A_c\,,\qquad
C^d_{[ab}C^n_{c]d}+\partial_iC^n_{[ab}A_{c]}^i=0\,,
\end{equation}
where the values $A_a=A_a^i(\phi)\frac{\partial}{\partial \phi^i}$
are viewed as vector fields on the $\phi$-space and the square
brackets stand for antisymmetrization of the indices enclosed. For
linearly independent vector fields $A_a$ the second equation in
(\ref{CCC}) follows from the first one by the Jacobi identity for
the Lie bracket of vector fields. From the viewpoint of the
target-space geometry, relations (\ref{CCC}) define a Lie
algebroid with anchor $A$.

The local BRST cohomology in the theories of type (\ref{Q}) has
been recently studied in \cite{BG}.  It was shown that under
certain assumptions this cohomology is isomorphic to the
$Q$-cohomology on the target space $\mathcal{M}$. In particular,
every Lagrange anchor was proven to be equivalent to the anchor
without space-time derivatives of fields and sources (the
algebraic Lagrange anchor). The proof assumed implicitly that the
order of the derivatives is bounded from above for all fields and
sources. This assumption does not follow from the fact that the
order of derivatives is finite for every single field or source,
whenever the number of fields is infinite. Hence, only two
possibilities are admissible: (i) the Lagrange anchor
(\ref{USF-on-An}) is equivalent to some algebraic anchor or (ii)
the equivalence class of the Lagrange anchor (\ref{USF-on-An})
does not have any representative with bounded order of
derivatives. Below, we prove that the on-shell unfolded
representation does not admit an algebraic Lagrange anchor
whenever $d>3$. This means that the unbounded order of derivatives
in (\ref{USF-on-An})-(\ref{GON3}) has no alternative.

\begin{prop}\label{T61}
The system of differential equations (\ref{FDA}) admits no
algebraic Lagrange anchor (except zero) whenever $d >3$.
\end{prop}
\begin{proof} Let $x^\mu$ be local coordinates on $M$.
Denote by $\bar\phi_i$ and $\bar\theta _a$ the conjugate sources.
It is convenient to think of $\bar\phi_i$ and $\bar\theta _a$ as
the sets of scalar and vector fields on $M$. The canonical
symplectic form on the phase space of fields and sources reads
\begin{equation}\label{}
\Omega=\int_M v(\delta\phi^i\wedge\delta\bar\phi_i+\delta
\theta^a_\mu\wedge \delta \bar\theta_a^\mu)\,,
\end{equation}
where $v$ is a suitable volume form on $M$ and $\wedge$ stands for
the exterior product of variational differentials. The most
general SD extension of equations (\ref{FDA}) by means of an
algebraic Lagrange anchor looks like
\begin{equation}\label{FFAn}\begin{array}{l}
\displaystyle  d \phi^i-A^i_a\theta^a+\left(U_{\mu\nu}^{ia}
\bar\theta^{\nu}_{a} +Y_{\mu}^{ij}\bar\phi_{j}\right)dx^\mu+\cdots
=0\,,\\[5mm]
\displaystyle d\theta^a-C^a _{bc}\theta^b\wedge\theta^c+\left(
V_{\mu\nu\lambda}^{ab}\bar\theta
_b^{\lambda}+W_{\mu\nu}^{ai}\bar\phi_i \right)dx^\mu\wedge dx^\nu
+ \cdots =0\,,
\end{array}\end{equation}
where the structure  functions $U$, $Y$, $V$,  $W$ depend on
$\phi^i$ and $\theta^a_\mu$, and the dots stand for possible terms
of higher degrees in sources. Substituting the ansatz (\ref{FFAn})
to the defining relations (\ref{anchor}) for the Lagrange anchor,
we find that a necessary condition for these relations to be
satisfied is that the following differential operators have to
vanish:
\begin{equation}\label{LTB1}
Y_{\mu}^{ij}\partial^{}_{\nu}+ Y_{\nu}^{ji}\partial^{}_{\mu}=0\,,
\end{equation}
\begin{equation}\label{LTB2}
W{}_{\mu\nu}^{ai}\partial^{}_{\tau}+
U_{\tau[\mu}^{ia}\partial^{}_{\nu]}=0\,,
\end{equation}
\begin{equation}\label{LTB3}
V_{\mu\nu[\tau}^{ab}\partial^{}_{\sigma]}+
V_{\tau\sigma[\mu}^{ba}\partial^{}_{\nu]}=0\,.
\end{equation}
For $d>3$ this is only possible if all the structure functions are
equal to zero. Indeed, setting $\mu\neq\nu$ in (\ref{LTB1}), we
conclude that the first and second summands vanish separately,
therefore $Y^{ij}_\mu=0$. For $\tau\neq\mu$ and $\tau\neq\nu$
equation (\ref{LTB2}) yields $W_{\mu\nu}^{ai}=0$ for all $\mu$,
$\nu$, and hence $U_{\tau\mu}^{ia}=0$. Finally, we let the indices
$\mu,\nu,\tau,\sigma$ to be pair-wise different. Then it follows
from (\ref{LTB3}) that the corresponding structure functions
$V_{\mu\nu\tau}^{ab}$ are equal to zero. If now two space-time
indices coincide (e.g. $\mu=\tau$), then we can always choose
$\nu\neq\mu$ and $\nu\neq\sigma$, which yields
$V_{\mu\sigma\mu}^{ab}=0$ for all $\mu,\sigma$. Thus, for all
values of indices $V_{\mu\nu\tau}^{ab}=0$. As is seen the crucial
point of our argumentation is the possibility to choose four
different values for the space-time indices, which is only
possible if $d>3$.
\end{proof}

As an immediate consequence of Proposition \ref{T61} we have that
the \textit{on-shell} unfolded formulation for the scalar field
theory (\ref{ew}), (\ref{phi}) admits no algebraic Lagrange
anchor. The analysis above can be easily extended to the FDAs
generated by forms of arbitrary degrees.

\begin{prop}\label{p2} Let $p$ be the highest degree of the forms entering equations (\ref{Q}).
Then these equations admit no algebraic Lagrange  anchor  whenever
$d \geq 2p+2$.
\end{prop}

Notice that the low bound for $\dim M$ established by Proposition
\ref{p2} is sharp. For if $p< d < 2p+2$, we have the sequence of
Lagrangian models
$$
S=\int_M B\wedge dH
$$
for the $p$-form $H$ and the $(d-p-1)$-form $B$. The corresponding
equations of motion,
$$
dH=0\,,\qquad dB=0\,,
$$
being Lagrangian, admit the canonical Lagrange anchor (\ref{CLA}),
which is obviously algebraic.

\section{Concluding remarks}

The unfolded formalism has many remarkable properties, though the
general unfolded field equations are not Lagrangian.  In the
standard unfolded form, even the scalar field equations do not
have any Lagrangian. The Lagrange anchor can exist, however, for
non-Lagrangian dynamics. Any nontrivial Lagrange anchor gives a
quantization of the dynamics and can connect symmetries to
conserved currents. So, the existence of the Lagrange anchor can
compensate for the most important disadvantages of the unfolded
formalism related to the non-variational structure of the field
equations. In this paper, we studied the structure of Lagrange
anchors admitted by the unfolded formalism.

As we have seen, an important feature of the Lagrange anchor in
the unfolded formalism is that it is a differential operator, in
general. For every single field, the order of the operator is
finite. In the unfolded formalism, however, the set of fields is
infinite. For example, in the scalar field theory, the set of
fields includes the original scalar field together with its
infinite jet prolongation. The order of derivatives is growing in
the Lagrange anchor with the order of corresponding jet. In this
sense the order of derivatives is unbounded. We also see that the
equivalence class of the Lagrange anchor does not include any
representative without derivatives.

The explicit construction of the Lagrange anchor for the unfolded
dynamics of scalar field can be basically repeated for any other
model admitting an equivalent formulation in terms of the finite
set of fields with any Lagrange anchor, canonical or not. Some
interesting field equations are formulated in the unfolded form
from the outset, not being equivalent to any local field theory
with a finite set of fields. An important example of this type is
provided by the interacting higher-spin fields \cite{V1},
\cite{V2}, \cite{BCIV}. Our construction can be applied, in
principle, for finding the Lagrange anchor even in this case. The
free limit of the higher-spin equations is known to be equivalent
to a Lagrangian field theory with a finite set of fields for every
spin. Hence, our procedure is applicable for the unfolded
equations of higher spin free fields, and it should lead to a
nontrivial Lagrange anchor. The Lagrange anchor for the
interacting equations can be sought for from the basic relation
(\ref{anchor}) by a perturbation theory in the interaction
constant, but the existence is not guaranteed. Obstructions can
appear, in principle, to the existence of the perturbative
solution for the Lagrange anchor. If any obstruction appeared, it
could have sense to identify that, because it would mean that the
interacting higher spin fields cannot be quantized in the way
consistent with the free limit. If no obstructions appeared, the
Lagrange anchor would make possible the quantization of the
interacting higher-spin fields, that remains an unsolved problem
for many years.

\vspace{0.2cm}

\noindent\textbf{Acknowledgements.} We appreciate the valuable
help from E. Skvortsov in the issues related to the on-shell
unfolded formalism. We are thankful to N. Boulanger and P. Sundell
for discussions on various topics addressed in this paper. The
research was supported by the Federal Targeted Program under the
state contracts 02.740.11.0238, P1337, P22 and by the RFBR grant
09-02-00723-a. Hospitality of the Erwin Schr\"{o}dinger Institute
for Mathematical Physics (Vienna) is acknowledged.

\end{document}